\documentclass[12pt,a4paper]{article}
\usepackage{bbm}
\usepackage{stmaryrd}
\usepackage{amsmath}
\usepackage{amssymb}
\usepackage{hyperref}
\usepackage[mathcal]{euscript}

\newcounter{thMM}[section]
\newcounter{leMM}[section]
\newcounter{deFF}[section]
\newcounter{exMP}[section]
\newcounter{prOP}[section]
\newcounter{coRR}[section]
\newcounter{coexMP}[section]
\newcounter{remAR}[section]
\newenvironment{theorem}[1][Theorem]{\refstepcounter{thMM}\trivlist
   \item[{\bf #1~\thesection.\arabic{thMM}.}]\it\hskip3pt}{\endtrivlist}
\newenvironment{lemma}[1][Lemma]{\refstepcounter{leMM}\trivlist
   \item[{\bf #1~\thesection.\arabic{leMM}.}]\it\hskip3pt}{\endtrivlist}
\newenvironment{definition}[1][Definition]{\refstepcounter{deFF}\trivlist
   \item[{\bf #1~\thesection.\arabic{deFF}.}]\rm\hskip3pt}{\endtrivlist}

\newenvironment{proposition}[1][Proposition]{\refstepcounter{prOP}\trivlist
   \item[{\bf #1~\thesection.\arabic{prOP}.}]\it\hskip3pt}{\endtrivlist}
\newenvironment{corollary}[1][Corollary]{\refstepcounter{coRR}\trivlist
   \item[{\bf #1~\thesection.\arabic{coRR}.}]\it\hskip3pt}{\endtrivlist}

\newenvironment{remark}[1][Remark]{\refstepcounter{remAR}\trivlist
\item[{\bf #1~\thesection.\arabic{remAR}.}]\rm\hskip3pt}{\endtrivlist}

\newenvironment{proof}[1][Proof]{\begin{trivlist}
\item[\hskip \labelsep {\bfseries #1}]}{\end{trivlist}}

\newcommand{\SN}[1]{\llbracket   #1 \rrbracket}
\newcommand{\proofend}{\flushright $\square$}
\newcommand{\Q}{\mathcal{Q}}
\newcommand{\Id}{\mathbbmss{1}}

\newcommand{\D}{\mathbbmss{D}}

\DeclareMathOperator{\Vect}{Vect}

\DeclareMathOperator{\w}{w}
\DeclareMathOperator{\proj}{proj}

\numberwithin{equation}{section}

\begin{document}
\author{Andrew James Bruce  \\ 
\newline \small{\emph{email:} \texttt{andrewjamesbruce@googlemail.com}  } }
\date{\today}
\title{Jacobi algebroids and quasi Q-manifolds}
\maketitle

\begin{abstract}
We reformulate the notion of a Jacobi algebroid in terms of weighted odd Jacobi brackets. We then show how a Jacobi algebroid can be understood in terms of  a  kind of \emph{curved}  Q-manifold. In particular  the homological condition on the odd vector field is \emph{deformed} in a very specific way. This leads to the notion of a quasi Q-manifold.
\end{abstract}
\begin{small}
\textbf{MSC 2010}: 17B70; 53D10; 53D17; 58A50.\\
\textbf{Keywords}: Supermanifolds,  Jacobi manifolds, Lie algebroids, contact structures, Q-manifolds.
\end{small}

\section{Introduction}\label{sec:Introduction}

Jacobi algebroids were first introduced by Iglesias and  Marrero \cite{Iglesias2001} under the name of generalised Lie algebroids. Such structures were the \emph{revisited} by Grabowski and Marmo \cite{Grabowski2001}. Jacobi algebroids represent a nice generalistion of the concept of a Lie algebroid. Indeed, Jacobi algebroids can be understood as Lie algebroids in the presence of a 1-cocycle \cite{Iglesias2001}. Recall the notion of a Lie algebroid as a vector bundle $E \rightarrow M$ equipped with a Lie bracket on the sections $[\bullet, \bullet]: \Gamma(E) \otimes \Gamma(E) \rightarrow \Gamma(E)$ together with an anchor $a: \Gamma(E) \rightarrow \Gamma(TM)$ that satisfy the Leibniz rule

\begin{equation}
\nonumber [u,fv] = a(u)[f] \: v  +  (-1)^{\widetilde{u}\widetilde{f}} f [u,v],
\end{equation}

where $u,v \in \Gamma(E)$ and $f \in C^{\infty}(M)$. The Leibniz rule implies that the anchor is actually a Lie algebra morphism: $a\left([u,v]\right) = [a(u), a(v)]$.  A Lie algebroid  can also  be understood in terms of:

\begin{enumerate}
\item a homological vector field of weight one on the total space of $\Pi E$. \label{homvect}
\item a weight minus one Poisson structure on the total space of $E^{*}$.
\item a weight minus one Schouten structure on the total space of $\Pi E^{*}$. \label{schouten}
\end{enumerate}

Here $\Pi$ is the parity reversion functor which shifts the Grassmann parity of the fibre coordinates, but does not effect the base coordinates. Also the parity reversion functor does not effect the assignment of the weights.\\

At this point we must remark that we will be working the the category of graded manifolds. That is  we will be working with supermanifolds equipped with a privileged class of atlases where the coordinates are assigned weights taking values in $\mathbbmss{Z}$ and the coordinate transformations are polynomial in coordinates with nonzero weights respecting the  weight. Generally the weight  will be independent of the Grassmann parity. Moreover, any sign factors that arise will be due to the Grassmann parity and we do not include any possible extra signs due to the weight. In simpler terms, we have a manifold equipped with a  distinguished class of charts and diffeomorphisms between them respecting the $\mathbbmss{Z}_{2}$-grading as well as the additional $\mathbbmss{Z}$-grading. These gradings then pass over to geometric objects (tensor and tensor-like objects) on graded manifolds. For further details about graded manifolds one can consult  \cite{Grabowski2009,Roytenberg:2001,Voronov:2001qf}.\\

Jacobi algebroids c.f. \cite{Grabowski2001,Iglesias2001} are understood in terms of a weight minus one even Jacobi bracket on $C^{\infty}(E^{*})$ for a given vector bundle $E \rightarrow M$. By  an even Jacobi bracket we mean a Poisson-like bracket in which the Leibnitz rule is weakened in  very specific way. We say more about this shortly.\\

In this paper we address the ``odd" approach to Jacobi algebroids generalising \ref{homvect}. and \ref{schouten}. of the above list mimicking the odd-super constructions related to Lie algebroids. The description of Lie algebroids in terms of homological vector fields is due to Va$\breve{\textrm{{\i}}}$ntrob \cite{Vaintrob:1997}. The deep links between Poisson geometry and Lie algebroids can be traced back to Coste, Dazord \& Weinstein \cite{Coste1987}.  The approach employed here is inherently \emph{super} as we work in the category of supermanifolds. We will assume basic knowledge of Lie algebroids and elementary  knowledge of supermanifolds.    \\

We continue this section with a brief statement of preliminaries needed for the rest of this paper. In \S\ref{sec:main constructions} we present our main constructions. In \S\ref{sec:from jacobi algebroids} we restate the association of Jacobi algebroids and lie algebroids in the presence of a 1-cocycle in language appropriate for the constructions of the previous section. In \S\ref{sec:shoutenisation} we show how to ``Schoutenise"  odd Jacobi structures  by extending the manifolds and thus build a \emph{larger} Lie algebroid directly from a Jacobi algebroid. We present a canonical example of a Jacobi algebroid in the form an odd contact manifold in \S\ref{sec:odd contact}. We round up with a few concluding remarks in \S\ref{sec:concluding}.  An appendix on the double vector bundle morphisms used in this paper is included.  \\

\begin{remark}
Antunes and Laurent--Gengoux \cite{Antunes2011} studied Jacobi structures using the supergeometric formulation of \emph{Buttin's big bracket}. They use this formulism to efficiently describe Jacobi algebroids and Jacobi bialgebroids. The supergeometric approach of the big bracket is different, but certainly related to that presented here.
\end{remark}

\noindent \textbf{Preliminaries} \\
All vector spaces and algebras will be $\mathbb{Z}_{2}$-graded.   We will generally  omit the prefix \emph{super}. By \emph{manifold} we will mean a \emph{smooth supermanifold}. We denote the Grassmann parity of an object by \emph{tilde}: $\widetilde{A} \in \mathbb{Z}_{2}$. By \emph{even} or \emph{odd} we will be referring explicitly to the Grassmann parity. \\

 A \emph{Poisson} $(\varepsilon = 0)$  or \emph{Schouten} $(\varepsilon = 1)$ \emph{algebra} is understood as a vector space $\mathbb{A}$ with a bilinear associative multiplication and a bilinear operation (``a bracket") $\{\bullet , \bullet\}_{\varepsilon}: \mathbb{A}  \otimes \mathbb{A} \rightarrow \mathbb{A}$ such that:
\begin{list}{}
\item \textbf{Grading} $\widetilde{\{a,b \}_{\varepsilon}} = \widetilde{a} + \widetilde{b} + \varepsilon$
\item \textbf{Skewsymmetry} $\{a,b\}_{\varepsilon} = -(-1)^{(\tilde{a}+ \varepsilon)(\tilde{b}+ \varepsilon)} \{b,a \}_{\varepsilon}$
\item \textbf{Jacobi Identity} $\displaystyle\sum\limits_{\textnormal{cyclic}(a,b,c)} (-1)^{(\tilde{a}+ \varepsilon)(\tilde{c}+ \varepsilon)}\{a,\{b,c\}_{\varepsilon}  \}_{\varepsilon}= 0$
\item \textbf{Leibniz Rule} $\{a,bc \}_{\varepsilon} = \{a,b \}_{\varepsilon}c + (-1)^{(\tilde{a} + \varepsilon)\tilde{b}} b \{a,c \}_{\varepsilon}$
\end{list} \vspace{10pt}
for all homogenous elements $a,b,c \in \mathbb{A}$.\\

If the Leibniz rule does not hold identically, but has an ``anomaly"  term as
\begin{equation}\nonumber
\{a,bc \}_{\varepsilon} = \{a,b \}_{\varepsilon}c + (-1)^{(\tilde{a} + \varepsilon)\tilde{b}} b \{a,c \}_{\varepsilon} - \{a ,\Id  \} bc,
\end{equation}

then we have  an \emph{even} ($\epsilon = 0)$ or \emph{odd} ($\epsilon = 1)$ \emph{Jacobi algebra}. Note that for Jacobi algebras  $\{a,\bullet\}_{\varepsilon}$ is a first order differential operator and not a vector field as in the Poisson/Schouten case. The theory of even Jacobi brackets on classical manifolds goes back to Lichnerowicz \cite{Lichnerowicz1977}, who also first introduced the notion of  a Poisson manifold. The analogous constructions of odd Jacobi brackets on supermanifolds was recently explored by the author. We will direct the reader to \cite{Bruce2011} for further details as we will recall the basic elements of the theory as needed.\\

A manifold $M$ such that $C^{\infty}(M)$ is a Poisson/Schouten algebra is known as a \emph{Poisson/Schouten manifold}. In particular the cotangent of a manifold comes equipped with a canonical Poisson structure.\\

Let us employ   natural local coordinates $(x^{A}, p_{A})$ on $T^{*}M$, with $\widetilde{x}^{A} = \widetilde{A}$ and $\widetilde{p}_{A} = \widetilde{A}$. Local diffeomorphisms on $M$ induce vector  bundle automorphism on $T^{*}M$ of the form
\begin{equation}\nonumber
\overline{x}^{A} = \overline{x}^{A}(x), \hspace{30pt} \overline{p}_{A}  = \left(\frac{\partial x^{B}}{\partial \overline{x}^{A}}\right)p_{B}.
\end{equation}

We will in effect use the local description as a \emph{natural vector bundle} to define the cotangent bundle of a supermanifold.  The canonical Poisson bracket on the cotangent is given by

\begin{equation}\nonumber
\{ F,G \} = (-1)^{\widetilde{A} \widetilde{F} + \widetilde{A}} \frac{\partial F}{\partial p_{A}}\frac{\partial G}{\partial x^{A}} - (-1)^{\widetilde{A}\widetilde{F}}\frac{\partial  F}{\partial x^{A}} \frac{\partial G}{\partial p_{A}}.
\end{equation}

\begin{definition}
The triple $(M, \D, q)$ with  $M$ being a manifold,  $\D \in \Vect(M)$ an odd vector field and  $q \in C^{\infty}(M)$ an odd function such that:

\begin{equation}
  \D^{2} = \frac{1}{2}[\D, \D] =  q \D, \hspace{15pt}\textnormal{and}\hspace{15pt}  \D[q] =0,
\end{equation}

shall be called a \textbf{quasi Q-manifold}. The vector field $\D$ shall be known as an \textbf{almost homological vector field}. The odd function $q$ shall be known as the \textbf{curving function}.
\end{definition}

\begin{definition}
Let $(M, \D_{M}, q_{M})$ and $(N, \D_{N}, q_{N})$ be quasi Q-manifolds and let $\phi: M \rightarrow N$ be  a smooth map. Then $\phi$ is said to be a \textbf{morphism of quasi Q-manifolds} if and only if
\begin{enumerate}
\item $\D_{M}\left( \phi^{*}f \right) = \phi^{*}\left(\D_{N}f \right)$ for all $f \in C^{\infty}(N)$. That is the almost homological vector fields are $\phi$-related.
\item $\phi^{*}q_{N} = q_{M}$. That is the curving functions match.
\end{enumerate}
\end{definition}
Quasi Q-manifolds and their morphisms form a category. Also note that if $q=0$ then we have the category of \emph{Q-manifolds} and $\D$ is a \emph{homological vector field}, \cite{Alexandrov:1995kv}. The curving functions ``measure" the failure of the homological condition of $\D$ and thus represent a kind of ``curvature". The other extreme is to set $\D =0$ and then keep $q$ as some distinguished odd function. For example, one could consider \emph{(higher) Schouten manifolds} as examples of quasi Q-manifolds. The far extreme is the trivial structure  of $\D =0$ and $q=0$ and we recover the full category of supermanifolds.

\section{Main constructions}\label{sec:main constructions}

In this section we propose a definition of a Jacobi algebroid in terms of an odd Jacobi structure on the total space of $\Pi E^{*}$, given a vector bundle $E \rightarrow M$. It will turn out that this definition is equivalent to that given by Grabowski \& Marmo \cite{Grabowski2001} (also see Iglesias \& Marrero \cite{Iglesias2001}). We postpone the details of this equivalence to the next section and take the following definition as the starting point of this work.

\begin{definition}
A vector bundle $E \rightarrow M$ is said to have the structure of a \textbf{Jacobi algebroid} if and only if the total space of $\Pi E^{*}$ comes equipped with a weight minus one odd Jacobi structure.
\end{definition}

Recall that an odd Jacobi structure on a manifold is a pair of odd functions on the total space of the cotangent bundle quadratic and linear in the fibre coordinates together with a series of conditions expressed in terms of the canonical Poisson bracket. Let us employ natural local coordinates $(x^{A}, \eta_{\alpha}, p_{A}, \pi^{\alpha})$ on the total space of $T^{*}(\Pi E^{*})$. The weight is assigned  as $\w(x^{A}) = 0$, $\w(p_{A})=0$, $\w(\eta_{\alpha}) = +1$ and $\w(\pi^{\alpha}) = -1$.  This is the \emph{natural weight} associated with the vector bundle structure $E^{*} \rightarrow M$. The parity of the coordinates is given by $ \widetilde{x}^{A}=  \widetilde{A}$, $\widetilde{\eta}_{\alpha}= (\widetilde{\alpha} +1)$, $\widetilde{p}_{A}= \widetilde{A}$ and  $\widetilde{\pi}^{\alpha} =  (\widetilde{\alpha}+1)$. In these natural local coordinates the odd Jacobi structure is given by

\begin{eqnarray}
S &=&(-1)^{\widetilde{\alpha}}\pi^{\alpha}Q_{\alpha}^{A}(x)p_{A}+ (-1)^{\widetilde{\alpha} + \widetilde{\beta}}\frac{1}{2}\pi^{\alpha}\pi^{\beta}Q_{\beta \alpha}^{\gamma}\eta_{\gamma},\\
\nonumber  \Q &=& \pi^{\alpha}Q_{\alpha}(x),
\end{eqnarray}

which are both functions on the total space of $T^{*}(\Pi E^{*})$. The notation and the sign factors employed  make clear the relation with Lie algebroids.   \\

This structure satisfies the conditions:
\begin{enumerate}
\item $\{\Q,\Q  \}_{T^{*}(\Pi E^{*})} =0$. \label{homcond}
\item $\{ \Q, S \}_{T^{*}(\Pi E^{*})} = 0$.
\item $\{S, S  \}_{T^{*}(\Pi E^{*})} =  - 2 \Q S$.
\end{enumerate}

Setting $\Q =0$ means that $S$ is a Schouten structure and thus we have a genuine Lie algebroid. See \cite{Bruce2011} for further details of more general odd Jacobi structures over manifolds and how the above conditions are required in order to build odd Jacobi algebras. Note that due to the fact that the function $\Q$ does not contain conjugate variables the condition \ref{homcond}. is automatically satisfied.  This is not generally the case and typically \ref{homcond}. will be a non-trivial condition.\\

The Jacobi algebroid structure on the vector bundle $E \rightarrow M$ is directly equivalent to the existence of a weight minus one odd Jacobi bracket on $C^{\infty}(\Pi E^{*})$. That is the algebra of  ``multivector fields" comes equipped with the structure of an odd Jacobi algebra \emph{viz}

\begin{equation}
\nonumber \SN{X,Y}_{E}  = (-1)^{\widetilde{X}+1}\{  \{S, X   \}_{T^{*}(\Pi E^{*})} , Y\}_{T^{*}(\Pi E^{*})}- (-1)^{\widetilde{X}+1} \{ \Q, XY \}_{T^{*}(\Pi E^{*})},
\end{equation}

with $X,Y \in C^{\infty}(\Pi E^{*})$.\\

In natural local coordinates this bracket is given by

\begin{eqnarray}\nonumber
\SN{X,Y}_{E} &=& Q_{\alpha}^{A}\left((-1)^{(\widetilde{X}+ \widetilde{\alpha}+1)(\widetilde{A}+1) } \frac{\partial X}{\partial \eta_{\alpha}}\frac{\partial Y}{\partial x^{A}}  - (-1)^{(\widetilde{X}+1)\widetilde{\alpha}} \frac{\partial X}{\partial x^{A}}\frac{\partial Y}{\partial \eta_{\alpha}}\right)\\
\nonumber &-& (-1)^{(\widetilde{X}+1)\widetilde{\alpha}+ \widetilde{\beta}}Q_{\alpha \beta}^{\gamma}\eta_{\gamma}\frac{\partial X}{\partial\eta_{\beta}}\frac{\partial Y}{\partial \eta_{\alpha}}\\
\nonumber &+&(-1)^{\widetilde{X}}Q_{\alpha}\frac{\partial X}{\partial \eta_{\alpha}} Y  + X Q_{\alpha}\frac{\partial Y}{\partial \eta_{\alpha}}.
\end{eqnarray}

Where $X = X(x, \eta) =  X(x) + X^{\alpha}(x) \eta_{\alpha} + \frac{1}{2!}X^{\alpha \beta}(x) \eta_{\beta}\eta_{\alpha} + \cdots$ \emph{etc}. The above odd Jacobi bracket is the natural generalisation of the weight minus one Schouten bracket associated with a Lie algebroid, which itself is a generalisation of the Schouten--Nijenhuis bracket between multivector fields over a manifold.\\

\begin{theorem}
The existence of a Jacobi algebroid structure on the vector bundle  $E \rightarrow M$ is equivalent to $\Pi E$ being  a weight one quasi Q-manifold.
\end{theorem}

\begin{proof}
Recall that the canonical double vector bundle morphism

\begin{equation}\nonumber
T^{*}(\Pi E^{*}) \stackrel{R}{\longrightarrow} T^{*}(\Pi E),
\end{equation}

is a symplectomorphism between the respective canonical symplectic structures.  We place details of this morphism in  Appendix(\ref{A1}). Thus we can \emph{move} the odd Jacobi structure from $\Pi E^{*}$ to $\Pi E$. However the resulting structure over $\Pi E$ will not be a genuine odd Jacobi structure as the degree in momenta (fibre coordinates of the cotangents) is not conserved under the canonical double vector bundle morphism.\\

Let us employ natural local coordinates $(x^{A}, \xi^{\alpha}, p_{A}, \pi_{\alpha})$ on $T^{*}(\Pi E)$. The weight of the coordinates is assigned as $\w(\xi^{\alpha})=-1$ and $\w(\pi_{\alpha})= +1$. The parities are $\widetilde{\xi}^{\alpha} = \widetilde{\pi}_{\alpha}= (\widetilde{\alpha}+1)$.  Then the canonical double vector bundle morphism is given by

\begin{equation}
\nonumber R^{*}\left( \pi_{\alpha} \right) =  \eta_{\alpha}, \hspace{25pt}  R^{*}\left(\xi^{\alpha}  \right)  = (-1)^{\widetilde{\alpha}} \pi^{\alpha}.
\end{equation}

Then let us consider

\begin{eqnarray}
\hat{S} := (R^{-1})^{*}S  &=& \xi^{\alpha}Q_{\alpha}^{A}(x) p_{A}+ \frac{1}{2}\xi^{\alpha}\xi^{\beta}Q_{\beta \alpha}^{\gamma}(x)\eta_{\gamma},\\
\nonumber \hat{\Q} := (R^{-1})^{*}\Q &=& (-1)^{\widetilde{\alpha}}\xi^{\alpha}Q_{\alpha}(x),
\end{eqnarray}

both of which are functions on the total space of $T^{*}(\Pi E)$. Note that the function $\hat{S}$ is now linear in moneta and that $\hat{\Q}$ is independent of momenta.  As $R$ is a symplectomorphism we naturally have

\begin{equation}\nonumber
\{ \hat{S}, \hat{S} \}_{T^{*}(\Pi E)} =  - 2 \hat{\Q}\hat{S},  \hspace{15pt}\textnormal{and}\hspace{15pt} \{\hat{\Q}, \hat{S} \}_{T^{*}(\Pi E)}=0.
\end{equation}

Then we can ``undo" the symbol map which gives an odd vector field on $\Pi E$ and an odd function linear in the fibre coordinate:

\begin{eqnarray}
\hat{S}&\longrightarrow& \D = \xi^{\alpha}Q_{\alpha}^{A}(x) \frac{\partial}{\partial x^{A}}+ \frac{1}{2}\xi^{\alpha}\xi^{\beta}Q_{\beta \alpha}^{\gamma}(x) \frac{\partial}{\partial \xi^{\gamma}} \in \Vect(\Pi E),\\
\nonumber \hat{\Q} &\longrightarrow& q =  - (-1)^{\widetilde{\alpha}}\xi^{\alpha}Q_{\alpha}(x)  \in C^{\infty}(\Pi E).
\end{eqnarray}

Note the extra minus sign in the definition of $q$. As the symbol map takes commutators of vector fields to Poisson brackets etc., it is not hard to see that the conditions that $(S, \Q)$  be an odd Jacobi structure translates to  $\Pi E$ being  a quasi Q-manifold: \\
\begin{center}
\begin{tabular}{lcl}
 $[\D,\D] =  2 q\D$,  & and  & $\D[q] =0$.
\end{tabular}
\end{center}
The grading is with respect to the natural grading associated with the vector bundle structure $E \rightarrow M$. That is we assign the weight as $\bar{\w}(x^{A}) =0$ and $\bar{\w}(\xi^{\alpha})= 1$. Note that $\bar{\w} = - \w$.
\proofend
\end{proof}


\section{From Jacobi algebroids to Lie algebroids in the presence of  a 1-cocycle}\label{sec:from jacobi algebroids}

In this section we  in essence restate  Grabowski \& Marmo's theorem 5 of \cite{Grabowski2001} giving a one-to-one correspondence between Jacobi algebroids and Lie algebroids in the presence of a 1-cocycle.\\

As we are considering the total space $\Pi E$ to be a graded manifold we naturally have an Euler vector field, which counts the weight of objects via it's Lie derivative. In natural local coordinates the Euler vector field is given by $E = \xi^{\alpha}\frac{\partial}{\partial \xi^{\alpha}}$ as we have assigned weight $\bar{\w}(x) =0$ and $\bar{\w}(\xi)= 1$. A ``differential form"  $\omega \in C^{\infty}(\Pi E)$ is homogeneous and of  weight $p$ if $E(\omega)= p \omega$. In a similar way, a vector field, $V \in \Vect(\Pi E)$ is homogeneous and of weight $r$ if $[E,V] = r V$. The action of the Euler vector field can be extended to higher tensor objects, but we will have no call to use it in this work. In relation to Jacobi algebroids, we will be exclusively  interested in object of weight one. Such objects are invariant under the action of the Euler vector field, or in more classical language they are \emph{linear} objects.

\begin{proposition}\label{prop1}
Let $(\Pi E , \D, q)$ be the weight one quasi Q-manifold associated with a  Jacobi algebroid. Then
\begin{equation}\nonumber
Q :=  \D - q E,
\end{equation}
defines a homological vector field on $\Pi E$ of weight one and thus a Lie algebroid structure on $E \rightarrow M$. Furthermore, we have $Q(q)=0$ and thus we have a Lie algebroid in the presence of a 1-cocycle.
\end{proposition}

\begin{proof}
The weight conditions are clear from the definitions. We need to prove that $Q$ is homological. Explicitly
\begin{eqnarray}
\nonumber Q^{2}\omega &=& \D^{2}\omega + q E\left(q E(\omega) \right) - \D\left ( qE(\omega)\right)- qE\left( \D \omega \right)\\
\nonumber &=& \D^{2}\omega - q [E, \D]\omega\\
\nonumber &=& \D^{2}\omega - q \D \omega,
\end{eqnarray}
for any $\omega \in C^{\infty}(\Pi E)$. Then using the fact that we have a quasi Q-manifold gives
\begin{equation}
\nonumber Q^{2}=0.
\end{equation}
It is clear that $Q(q)=0$ and thus we have a 1-cocycle.
\proofend
\end{proof}

\begin{proposition}\label{prop2}
Let $(\Pi E, Q)$ be a Lie algebroid and let $\phi \in C^{\infty}(\Pi E)$ be an odd 1-cocycle, that is $E(\phi)=1$,  $Q(\phi) =0$ and $\widetilde{\phi}=1$. Then
\begin{equation}\nonumber
\left( \Pi E, \D = Q + \phi E, q = \phi \right),
\end{equation}
defines a quasi Q-manifold of weigh one, and thus a Jacobi algebroid.
\end{proposition}

\begin{proof}
The conditions on the weights is clear. Then via calculation we obtain
\begin{eqnarray}
\nonumber \D^{2}\omega &=& Q^{2}\omega + \phi E \left(\phi E \omega  \right) + Q \left( \phi E \omega\right) + \phi E \left( Q\omega \right)\\
\nonumber &=& \phi [E,Q]\omega  = \phi \left(Q + \phi E  \right)\omega  = \phi \D \omega.
\end{eqnarray}
The 1-cocycle condition implies $\D(\phi)=0$.
\proofend
\end{proof}

\begin{theorem}(\textbf{Grabowski--Marmo \cite{Grabowski2001}})
There is a one-to-one correspondence between Jacobi algebroids and Lie algebroids in the presence of an odd 1-cocycle.
\end{theorem}

We must again  remark that everything here is done in the category of supermanifolds and that we have both Grassmann even and odd cocycles. For the classical case where $E \rightarrow M$ is  in the category of pure even classical manifolds 1-cocycles are necessarily odd. Thus the above propositions and theorem include the classical structures.\\

For clarity let us examine the association of a Lie algebroid in the presence of a 1-cocycle with a Jacobi algebroid in natural local  coordinates. It is not hard to see that given $\D$ and $q$ we have

\begin{eqnarray}
Q &=& \xi^{\alpha} Q_{\alpha}^{A} \frac{\partial}{\partial x^{A}} + \frac{1}{2}\left(\xi^{\alpha}\xi^{\beta} Q_{\beta \alpha}^{\gamma} +  (-1)^{\widetilde{\alpha}} 2 \xi^{\alpha}Q_{\alpha} \xi^{\gamma}  \right)\frac{\partial}{\partial \xi^{\gamma}} \in \Vect(\Pi E),\\
\nonumber \phi &=& (-1)^{\widetilde{\alpha} +1}\xi^{\alpha}Q_{\alpha}.
\end{eqnarray}

By careful symmetrisation we see that building the Lie algebroid structure on $\Pi E$ associated with a Jacobi algebroid is essentially described by the replacement
\begin{equation}
\nonumber \D  \longrightarrow Q,
\end{equation}
viz
\begin{equation}
\nonumber Q_{\beta \alpha}^{\gamma} \longrightarrow Q_{\beta \alpha}^{\gamma} - (-1)^{\widetilde{\alpha}+ \widetilde{\beta}}\left(\delta_{\alpha}^{\:\: \gamma}Q_{\beta} + (-1)^{(\widetilde{\alpha}+1)(\widetilde{\beta}+1)} Q_{\alpha}\delta_{\beta}^{\:\: \gamma}  \right).
\end{equation}
One can then more-or-less read off the Lie bracket on the sections of $E$ and the anchor map $a : \Gamma(E) \rightarrow \Vect(M)$.  Picking a basis of sections $(s_{\alpha})$ for $\Gamma(E)$ and being intentionally slack with the signs we have

\begin{eqnarray}
\nonumber [s_{\alpha}, s_{\beta}] &=& \pm Q_{\alpha \beta}^{\gamma} s_{\gamma} \pm Q_{\alpha}s_{\beta} \pm s_{\alpha}Q_{\beta},\\
\nonumber a(s_{\alpha}) &=& \pm Q_{\alpha}^{A}\frac{\partial}{\partial x^{A}}.
\end{eqnarray}

Dual to this one can consider the associated Schouten structure which is given by

\begin{equation}
\bar{S} = (-1)^{\widetilde{\alpha}} \pi^{\alpha}Q_{\alpha}^{A}p_{A} + \frac{1}{2}\left( (-1)^{\widetilde{\alpha}+ \widetilde{\beta}}\pi^{\alpha}\pi^{\beta}Q_{\beta \alpha}^{\gamma} + (-1)^{\widetilde{\gamma}}2 \pi^{\alpha} Q_{\alpha}\pi^{\gamma} \right)\eta_{\gamma} \in C^{\infty}(T^{*}(\Pi E^{*})).
\end{equation}

Similarly, the 1-cocycle becomes $\bar{\phi} = - \pi^{\alpha}Q_{\alpha}$ and it is not hard to see that

\begin{eqnarray}
\nonumber \{\bar{S}, \bar{S}  \}_{T^{*}(\Pi E^{*})} &=& 0,\\
\nonumber \{ \bar{S} , \bar{\phi} \}_{T^{*}(\Pi E^{*})} &=& 0.
\end{eqnarray}

\begin{corollary}\label{cor:1}
If $\mathfrak{g}$ is a Lie algebra with a distinguished odd 1-cocycle $\phi \in C^{\infty}(\Pi \mathfrak{g})$ then $\Pi \mathfrak{g}^{*}$ is a (formal) odd Jacobi manifold.
\end{corollary}

In local coordinates we have $Q = \frac{1}{2}\xi^{\alpha}\xi^{\beta}Q_{\beta \alpha}^{\gamma} \frac{\partial}{\partial \xi^{\gamma}} \in \Vect(\Pi \mathfrak{g})$ which \emph{encodes}   the Lie algebra structure on $\mathfrak{g}$. The 1-cocycle is given by $\phi = (-1)^{\widetilde{\alpha}} \xi^{\alpha} Q_{\alpha}$, the sign is picked for convenience. Then the odd Jacobi structure on $\Pi \mathfrak{g}^{*}$ is given by

\begin{eqnarray}
\nonumber S &=& (-1)^{\widetilde{\alpha}+ \widetilde{\beta}}\frac{1}{2} \pi^{\alpha}\pi^{\beta}Q_{\beta \alpha}^{\gamma}\eta_{\gamma}+ (-1)^{\widetilde{\gamma}}\pi^{\alpha}Q_{\alpha} \pi^{\gamma} \eta_{\gamma},\\
\nonumber \Q &=& \pi^{\alpha}Q_{\alpha}.
\end{eqnarray}

The associated odd Jacobi brackets should be thought of  generalisation of the ``Lie--Schouten" bracket on $\Pi \mathfrak{g}^{*}$  \cite{Voronov:2001qf} in the presence of a 1-cocycle. Both these odd brackets are then considered as odd generalisations of the ``Lie--Poisson--Berezin--Kirillov" bracket on $\mathfrak{g}^{*}$.

\begin{corollary}\label{cor:2}
If $(\Pi E, Q)$ is a Lie algebroid, then $\Pi E^{*} \otimes \mathbbmss{R}^{0|1}$ is a Jacobi algebroid.
\end{corollary}

Let us employ natural local coordinates on $T^{*}(\Pi E^{*} \otimes \mathbbmss{R}^{0|1})$ which we denote as $(x^{A}, \eta_{\alpha}, \tau, p_{A}, \pi^{\alpha}, \pi)$.  The weight assigned to these extra coordinates is $\w(\tau) =1$ and $\w(\pi)=-1$.  In these local coordinates the weight minus one Jacobi structure is given by

\begin{eqnarray}
\nonumber S&=& (-1)^{\widetilde{\alpha}}\pi^{\alpha}Q_{\alpha}^{A} p_{A} + (-1)^{\widetilde{\alpha} + \widetilde{\beta}} \frac{1}{2}\pi^{\alpha}\pi^{\beta} Q_{\beta \alpha}^{\gamma}\eta_{\gamma} + \pi \pi^{\alpha}\eta_{\alpha},\\
\nonumber \mathcal{Q} &=& - \pi.
\end{eqnarray}

\noindent \textbf{Statement:} extending the fibres of the vector bundle $E \rightarrow M$ underlying a Lie algebroid by $\mathbbmss{R}$ allows one to directly construct a Jacobi algebroid structure on $\Pi E^{*} \otimes \mathbbmss{R}^{0|1}$.\\

Naturally the proceeding corollary includes Lie algebra as Lie algebroids over a point. Then, if $\mathfrak{g}$ is a Lie algebra one can \emph{extend} the vector space structure to $\mathfrak{g} \otimes \mathbbmss{R}$. Directly associated with this is the (formal) manifold $\Pi (\mathfrak{g}^{*} \otimes \mathbbmss{R})$ which comes with an odd Jacobi structure of weight minus one.

\begin{corollary}\label{cor:3}
Let $M$ be a manifold and $\mathbbmss{A}$ be a closed, odd one-form  (a flat Abelian connection). Then $\Pi TM$ can be made into quasi Q-manifold of weight one, or in other words, $\Pi T^{*}M$ can be considered as a Jacobi algebroid.
\end{corollary}

In natural local coordinates $(x^{A}, dx^{A})$ on $\Pi TM$, the quasi Q-manifold structure is given by

\begin{eqnarray}
\nonumber \D &=& d + \mathbbmss{A} \:E\\
\nonumber &=& dx^{A} \frac{\partial}{\partial x^{A}} + dx^{B}\mathbbmss{A}_{B} \: dx^{A} \frac{\partial}{\partial dx^{A}},\\
\nonumber q &=& \mathbbmss{A} = dx^{B}\mathbbmss{A}_{B}.
\end{eqnarray}

The weight here is simply assigned as $\w(x^{A})=0$ and $\w(dx^{A})=1$. Picking natural local coordinates $(x^{A}, x^{*}_{A}, p_{A}, p^{*}_{A})$ on $T^{*}(\Pi T^{*}M)$ allows us to write the corresponding odd Jacobi structure on $\Pi T^{*}M$ as

\begin{eqnarray}
\nonumber S &=& (-1)^{\widetilde{A}}p_{*}^{A}p_{A} + (-1)^{\widetilde{B}}p_{*}^{B}\mathbbmss{A}_{B}\: p_{*}^{A}x^{*}_{A},\\
\nonumber \mathcal{Q} &=& - p_{*}^{A}\mathbbmss{A}_{A}.
\end{eqnarray}

Note that the first term of the almost Schouten structure is the canonical Schouten structure on the anticotangent bundle.

\section{``Schoutenisation" and Lie algebroids}\label{sec:shoutenisation}

In this section we show that given arbitrary Jacobi algebroid one can extend the structure via a  ``Schoutenisation" process to construct a genuine Lie algebroid.  The constructions here mimic very closely the ``Poissonisation" of a classical Jacobi manifold.  Consider the manifold $T^{*}(\Pi E^{*} \otimes \mathbbmss{R})$ which we equip with local coordinates $(x^{A}, \eta_{\alpha}, t,p_{A}, \pi^{\alpha}, p)$. The weight we assign as:   \\

\begin{tabular}{|l ||l| }
\hline
$ \w(x^{A}) = 0$  & $\w(p_{A}) = 0$ \\
$ \w(\eta_{a})=1$  & $\w(\pi^{a}) =-1$\\
$\w(t)=0$ &   $\w(p)=0$\\
\hline
\end{tabular}\\

\begin{proposition}
Let $(\Pi E^{*}, S, \mathcal{Q})$ be a Jacobi algebroid. Then $\Pi E^{*} \otimes \mathbbmss{R}$ is a weight minus one Schouten manifold where the Schouten structure is given by
\begin{equation}
\bar{S} = e^{-t} \left( S - \mathcal{Q}p \right).
\end{equation}
\end{proposition}

\begin{proof}
Let us denote the canonical Poisson bracket on $T^{*}(\Pi E^{*} \otimes \mathbbmss{R})$ by $\{\bullet  ,\bullet\}$. Note then that we have the natural decomposition $\{ , \} = \{ , \}_{T^{*}(\Pi E^{*})} + \{ , \}_{T^{*}\mathbbmss{R}}$. It is then a straight forward exercise to take into account terms that contain conjugate variables and those that don't to show that
\begin{equation}
\nonumber \{ \bar{S}, \bar{S} \} = e^{-2t} \left(\{ S,S \}_{T^{*}(\Pi E^{*})} + 2 \mathcal{Q}S - 2p \{ S, \mathcal{Q}\}_{T^{*}(\Pi E^{*})}    \right). \end{equation}
Thus as the pair $(S, \mathcal{Q})$ define an odd Jacobi structure on $\Pi E^{*}$,  $\bar{S} \in T^{*}(\Pi E^{*} \otimes \mathbbmss{R})$ is a Schouten structure on $\Pi E^{*} \otimes \mathbbmss{R}$. The assignment of the weight follows directly from the definition.
\proofend
\end{proof}

\begin{remark}
Of course as the above proof does not rely on the graded structure of $\Pi E^{*}$ the ``Schoutenisation" process extends directly to arbitrary odd Jacobi manifolds.
\end{remark}

In natural local coordinates this Schouten structure is given by

\begin{equation}
\bar{S} = e^{-t} \left((-1)^{\widetilde{\alpha}} \pi^{\alpha}Q_{\alpha}^{A}p_{A} + (-1)^{\widetilde{\alpha}+ \widetilde{\beta}}\frac{1}{2}\pi^{\alpha}\pi^{\beta} Q_{\beta \alpha}^{\gamma}\eta_{\gamma} - \pi^{\alpha}Q_{\alpha}p \right).
\end{equation}

We need to understand the vector bundle structure in order to really identify the Lie algebroid structure. Given the weight assigned to the coordinates on $\Pi E^{*} \otimes \mathbbmss{R}$ the associate underlying (dual) vector bundle structure is $\proj^{*}E \longrightarrow M \otimes \mathbbmss{R}$. That is the pullback of $E \rightarrow M$ by   $\proj : M \otimes \mathbbmss{R} \rightarrow M$.

\begin{corollary}
If $\Pi E^{*}$ has the structure of a Jacobi algebroid then $\proj^{*}E \longrightarrow M \otimes \mathbbmss{R}$ is a Lie algebroid.
\end{corollary}

\noindent \textbf{Statement:} given a Jacobi algebroid  structure on $\Pi E^{*}$, one can extend the base space $M$ of the underlying vector bundle $E \longrightarrow M$ by $\mathbbmss{R}$ to directly construct a Lie algebroid.

\section{Odd contact manifolds and Jacobi algebroids}\label{sec:odd contact}

In this section we  show that the manifold $M :=  \Pi T^{*}N \otimes \mathbbmss{R}^{0|1}$ considered as an odd contact manifold provides a canonical example of a Jacobi algebroid than lends itself to the description in terms of odd Jacobi brackets. One should consider odd Jacobi manifolds as a specific example in the context of Corollary 3.\ref{cor:2}.\\

Let $N$ be a pure even classical manifold of dimension $n$. Consider the manifold $M :=  \Pi T^{*}N \otimes \mathbbmss{R}^{0|1}$ equipped with natural local coordinates $(x^{a}, x^{*}_{a}, \tau)$. The coordinates $x^{a}$ are even, while the other coordinates $x^{*}_{a}$ and $\tau$ are odd. The dimension of $M$ is $(n|n+1)$. The manifold $M$ comes equipped with an \emph{odd contact one form}, which is the even one form
\begin{equation}
\alpha = d \tau - x^{*}_{a}dx^{a}.
\end{equation}

It was shown in \cite{Bruce2011} that $M$ is an odd Jacobi manifold with the odd Jacobi structure being

\begin{equation}
S = p^{a}_{*}  \left( p_{a} + x^{*}_{a} \pi\right),  \hspace{30pt} \mathcal{Q} =  - \pi,
\end{equation}
where we have employed natural coordinates $(x^{a}, x^{*}_{a}, \tau, p_{a}, p_{*}^{a} ,\pi)$ on $T^{*}M$. Indeed this odd Jacobi structure is directly equivalent to the odd contact structure. Without details, both the odd contact and odd Jacobi structure on $M$ can be considered as the ``natural superisation" of the classical structures on $ \mathbbmss{R}^{3}$. Note that $\Pi T^{*}N$ comes equipped with a canonical Schouten (odd symplectic) structure, but $\Pi T^{*}N \otimes \mathbbmss{R}^{0|1}$ comes with a canonical odd Jacobi structure. \\

Let us  attach the weight to the local coordinates on $M$ as:\\

\begin{tabular}{|l ||l| }
\hline
$ \w(x^{a}) = 0$  & $\w(p_{a}) = 0$ \\
$ \w(x^{*}_{a})=1$  & $\w(p_{*}^{a}) =-1$\\
$\w(\tau)=1$ &   $\w(\pi)=-1$\\
\hline
\end{tabular}\\

This weight is the ``natural weight" with respect to the underling  vector bundle structure $T^{*}N \otimes \mathbbmss{R} \longrightarrow N$. With respect to this weight it is clear that the odd  Jacobi structure on $M := \Pi T^{*}N \otimes \mathbbmss{R}^{0|1}$ is of weight minus one and we thus have a Jacobi algebroid. \\

Now consider $M^{\star} := \Pi TN \otimes \mathbbmss{R}^{0|1}$ equipped with natural local coordinates $(x^{a}, \xi^{a}, \eta, p_{a}, \pi_{a}, \theta)$. The canonical double vector bundle morphism $R: T^{*}M \rightarrow T^{*}M^{\star}$ act on the coordinates as $R^{*}(\xi^{a}) = p_{*}^{a}$, $R^{*}(\eta) = \pi$, $R^{*}(\pi_{a})= x^{*}_{a}$ and $R^{*}(\theta)= \tau$. Then we can pull-back the odd Jacobi structure to give

\begin{eqnarray}
\nonumber \hat{S} &=& \xi^{a}(p_{a} + \pi_{a} \eta),\\
\nonumber \hat{\mathcal{Q}} &=& - \eta,
\end{eqnarray}

both of which are now functions on the total space of $M^{\star}$.  Then we can ``undo" the symbol (and after a little reordering) to produce

\begin{eqnarray}
\D &=& \xi^{a}\frac{\partial}{\partial x^{a}} + \eta \xi^{a}\frac{\partial }{\partial \xi^{a}},\\
\nonumber q &=& \eta.
\end{eqnarray}

Direct calculation confirms that $M^{\star} := \Pi TN \otimes \mathbbmss{R}^{0|1}$ is a quasi Q-manifold. We see that in light of Proposition  3.\autoref{prop2}. that the can consider $TN \otimes \mathbbmss{R} \rightarrow N$ as a Lie algebroid in the presence of a 1-cocycle. The de Rham differential on $N$ is the associated homological vector field and the 1-cocycle is identified with the ``odd time".\\

\begin{remark}
As this work was being completed, Mehta \cite{Mehta2011} established a one-to-one correspondence between Jacobi manifolds and degree 1 contact  $NQ$-manifolds. Mehta shows how to interpret the ``Poissonisation" of a Jacobi manifold as the ``symplectification" of the corresponding degree 1 contact $NQ$-manifold. There is no doubt that Mehta's results can be slightly reformulated to sit comfortably with the conventions used here: one would consider ``Shoutenisation" and ``symplectification" of odd contact structures.  This generalises  the  correspondence between Poisson manifolds and degree 1 symplectic $NQ$-manifolds, as established by Roytenberg \cite{Roytenberg:2001}.
\end{remark}

\section{Concluding remarks}\label{sec:concluding}
In this paper we \emph{define} Jacobi algebroids  in terms of an odd Jacobi structure on $\Pi E^{*}$ of weight minus one. That is the ``multivector fields" come equipped with an odd Jacobi bracket. For Lie algebroids the bracket between ``multivector fields" is a Schouten bracket, i.e. satisfies a strict Leibniz rule.  \\

This construction was then used to construct a weight one almost homological vector field.  That is the ``differential forms" come equipped with a kind of \emph{deformed} de Rham differential. Importantly we no longer  have a homological vector field as in the case of Lie algebroids, but rather the homological condition is weakened in a very specific way as to provide a quasi Q-manifold structure. As such Jacobi algebroids can be considered as very specific  examples of  \emph{skew algebroids}  \cite{Grabowski1999}, which are a kind of Lie algebroid in which the Jacobi identity is lost. If the corresponding anchor is  a Lie algebra morphism between sections of the vector bundle and vector fields over the base then we have the notion of an \emph{almost Lie algebroid} \cite{Leon2010}. Interest in algebroids without the Jacobi identity comes from nonholonomic mechanics, where skew algebroids provide a general geometric setting.\\

However, via a simple redefinition one can rephrase Jacobi algebroids in terms of Lie algebroids in the presence of a 1-cocycle, which are also known as generalised Lie algebroids. In doing so we recover, maybe up to conventions, the notion of a Jacobi algebroid in the sense of \cite{Grabowski2001,Iglesias2001}. \\

\noindent \emph{Questions that naturally follow from this work include:}
\begin{itemize}
\item Do Jacobi bialgebroids have an efficient description in terms of a compatible odd Jacobi structure and a quasi Q-structure?
\item Can one define non-linear Jacobi algebroids in terms of odd Jacobi structures over non-negatively graded supermanifolds? Are these naturally related to Voronov's  (\cite{Voronov2010}) non-linear Lie algebroids?
\item Is there a good notion of a Jacobi-$\infty$  algebroid, \emph{id est} via weakening the weight condition on the Jacobi structures but keeping the vector bundle structure? Or even just a Jacobi version of an $L_{\infty}$-algebra? How does the notion of $L_{\infty}$-algebroids (\cite{Bruce2011b}) fit in here? What about Lie algebroids in the presence of an $n$-cocycle? Clearly most of the constructions carry over to to the $n$-cocycle case but the weights of the objects will be inhomogeneous: what does this mean?
\end{itemize}

\section*{Acknowledgements} The author would like to thank  J. Grabowski and  R.A. Mehta for there  invaluable comments on  earlier versions of this work.

\appendix
\section*{Appendix}\label{appendix}
\section{Canonical double vector bundle morphisms}
\noindent For completeness we present the canonical double vector bundle morphisms used in this work. In particular we prove that the morphisms are symplectomorphisms. We describe vector bundles in terms of coordinates on their total spaces and the associated vector bundle automorphisms. Specifically we have:\\

\begin{tabular}{|l||l|}
\hline
$E \longrightarrow M$  &  $E^{*}\longrightarrow M$\\
$(x^{A}, e^{\alpha})  \mapsto (x^{A}) $ & $(x^{A}, e_{\alpha})  \mapsto (x^{A}) $\\
\hline
&\\
$\overline{x}^{A} = \overline{x}^{A}(x)$& $\overline{x}^{A} = \overline{x}^{A}(x)$\\
$\overline{e}^{\alpha} = e^{\beta}T_{\beta}^{\:\: \alpha}(x)$ & $\overline{e}_{\alpha} = \left( T^{-1} \right)_{\alpha}^{\:\: \beta}(x)$\\
\hline
\end{tabular}\\

Where $T_{\beta}^{\:\: \gamma} \left( T^{-1} \right)_{\gamma}^{\:\: \alpha} = \delta_{\beta}^{\:\: \alpha}$ \emph{etc}. We take  $\widetilde{e^{\alpha}} = \widetilde{e_{\alpha}}= \widetilde{\alpha}$.

\subsection{$T^{*}(\Pi E^{*})$ and $T^{*}(\Pi E)$}\label{A1}
\noindent Let us employ natural local coordinates:

\vspace{15pt}
\begin{tabular}{|l ||l| }
\hline
$ T^{*}(\Pi E^{*})$ & $(x^{A},\eta_{\alpha}, p_{A}, \pi^{\alpha} )$ \\
$  T^{*}(\Pi E)$  & $(x^{A},\xi^{\alpha}, p_{A}, \pi_{\alpha} )$\\
\hline
\end{tabular}\\

\noindent The Grassmann  parities are given by $ \widetilde{x}^{A}= \widetilde{p}_{A}= \widetilde{A}$, $\widetilde{\eta}_{\alpha}= \widetilde{\pi}^{\alpha}= \widetilde{\pi}_{\alpha}= \widetilde{\xi}^{\alpha} =  (\widetilde{\alpha}+1)$. The weights are assigned as $\w(x^{A}) =0$, $\w(\eta_{\alpha}) =1$, $\w(p_{A}) =0$, $\w(\pi^{\alpha}) =-1$ , $\w(\xi^{\alpha}) =-1$, $\w(\pi_{\alpha}) =1$.  The admissible changes of coordinates are:

\vspace{15pt}
\begin{tabular}{|l||l|}
\hline
& \\
$T^{*}(\Pi E^{*})$ &   $\overline{x}^{A}  =  \overline{x}^{A}(x)$, \hspace{5pt} $\overline{\eta}_{\alpha}  =   (T^{-1})_{\alpha}^{\:\: \beta}\eta_{\beta}$,\\
 & $\overline{p}_{A} = \left( \frac{\partial x^{B}}{\partial \overline{x}^{A}} \right)p_{B} + (-1)^{\widetilde{A}(\widetilde{\gamma}+ 1) + \widetilde{\delta}} \pi^{\delta}T_{\delta}^{\:\: \gamma} \left( \frac{\partial (T^{-1})_{\gamma}^{\:\: \alpha}}{\partial \overline{x}^{A}} \right)\eta_{\alpha}$,\\
 & $ \overline{\pi}^{\alpha} = (-1)^{\widetilde{\alpha} + \widetilde{\beta}}\pi^{\beta}T_{\beta}^{\:\: \alpha}$.\\
\hline
&\\
$T^{*}(\Pi E)$ & $ \overline{x}^{A}  =  \overline{x}^{A}(x)$, \hspace{5pt}$\overline{\xi}^{\alpha}  =   \xi^{\beta} T_{\beta}^{\:\: \alpha}$,\\
& $\overline{p}_{A} = \left( \frac{\partial x^{B}}{\partial \overline{x}^{A}} \right)p_{B} + (-1)^{\widetilde{A}(\widetilde{\gamma}+1)} \xi^{\delta}T_{\delta}^{\:\: \gamma} \left(\frac{\partial (T^{-1})_{\gamma}^{\:\: \alpha}}{\partial \overline{x}^{A}}  \right)\pi_{\alpha}$,\\
&  $\overline{\pi}_{\alpha} = (T^{-1})_{\alpha}^{\:\: \beta} \pi_{\beta}$.\\
\hline
\end{tabular}\\

\vspace{15pt}

\noindent There is canonical double vector bundle morphism $R: T^{*}(\Pi E^{*}) \rightarrow T^{*}(\Pi E )$ given in local coordinates as

\begin{equation}\nonumber
R^{*}(\pi_{\alpha}) = \eta_{\alpha},  \hspace{35pt} R^{*}(\xi^{\alpha}) = (-1)^{\widetilde{\alpha}}\pi^{\alpha}.
\end{equation}

\begin{lemma}\label{lemma1}
The canonical double vector bundle morphism $R: T^{*}(\Pi E^{*}) \rightarrow T^{*}(\Pi E )$ is a symplectomorphism.
\end{lemma}

\begin{proof}
The canonical even  symplectic structure on $T^{*}(\Pi E^{*})$ is given by $\omega_{T^{*}(\Pi E^{*})} =  dp_{A}dx^{A} + d\pi^{\alpha}d \eta_{\alpha}$ and on $T^{*}(\Pi E)$ is given by $\omega_{T^{*}(\Pi E)} = dp_{A}dx^{A} + d\pi_{\alpha}d\xi^{\alpha}$. Thus, $R^{*}\omega_{T^{*}(\Pi E )} = \omega_{T^{*}(\Pi E^{*})}$ and we see that $R$ is indeed a symplectomorphism.
\proofend
\end{proof}

\vfill
\begin{center}
Andrew James Bruce\\
\small{\emph{email:} \texttt{andrewjamesbruce@googlemail.com}  }
\end{center}
\end{document}